\documentclass[sn-mathphys]{sn-jnl}

\jyear{2023}%

\theoremstyle{thmstyleone}%
\newtheorem{theorem}{Theorem}
\newtheorem{proposition}[theorem]{Proposition}%

\theoremstyle{thmstyletwo}%

\theoremstyle{thmstylethree}%
\newtheorem{assumption}{Assumption}
\usepackage{mathtools}
\DeclarePairedDelimiter{\abs}{\lvert}{\rvert}

\begin{document}

\title[Basic concepts for the Kermack and McKendrick model with static heterogeneity]{Basic concepts for the Kermack and McKendrick model with static heterogeneity}


\author{\fnm{Hisashi} \sur{Inaba}}\email{inaba57@u-gakugei.ac.jp}



\affil*[1]{\orgdiv{Department of Education}, \orgname{Tokyo Gakugei University}, \orgaddress{\street{4-1-1 Nukuikita-machi}, \city{Koganei-shi}, \postcode{184-8501}, \state{Tokyo}, \country{Japan}}}




\abstract{ In this paper, we consider the infection-age-dependent Kermack--McKendrick model in which host individuals are distributed in a continuous state space.
To provide a mathematical foundation for the heterogeneous model, we develop a $L^1$-framework to formulate basic epidemiological concepts.
  First, we show the mathematical well-posedness of the basic model under appropriate conditions  allowing the unbounded parameters with non-compact domain. Next we define the basic reproduction number and prove the pandemic threshold theorem.  We then present a systematic procedure to compute the effective reproduction number and the herd immunity threshold. Finally we give some illustrative examples and concrete results by using the separable mixing assumption.   
}




\maketitle

\section{Introduction}\label{sec1}
Among the mathematical studies inspired by COVID-19, the Kermack-McKendrick epidemic model of 1927 \cite{Kermack1927} has played an important role.  In particular, it was recognized that to understand the complex dynamics of recurrent waves of epidemics, we need to extend the basic model to account for individual heterogeneity described by the continuous trait variable that reflects biological or social heterogeneity of individuals \cite{Britton2020, Diekmann2023, Gomes2022, Montalban2022, Neipel2020, Tkachenko2021a, Tkachenko2021b}.

As a special case of individual heterogeneity described by continuous trait variables, spatial extensions of the Kermack--McKendrick model have a long history.  Kendall \cite{Kendall1957} proposed a spatial extension of the Kermack--McKendrick model and stated his {\it Pandemic Threshold Theorem}.  Baily \cite{Bailey1975} defines a {\it pandemic} when the proportion of individuals contracting the disease, whatever the distance from the initial focus of infection is, is greater or equal to the root of the final size equation, provided that the basic reproduction number exceeds unity.
Diekmann \cite{Diekmann1978}, Thieme \cite{Thieme1977a, Thieme1977b},  Webb \cite{Webb1980, Webb1981}, Rass and Radcliffe \cite{Rass2003} and Inaba \cite{Inaba2014} extended Kendall's pandemic threshold result to the infection-age structured Kermack--McKendrick model with continuous individual heterogeneity.

On the other hand,  the complex behavior of COVID-19 has suggested us that the social and biological individual heterogeneity would play a most important role in the global dynamics of an outbreak.  Therefore, we need to develop a mathematical model that can take into account the social and biological individual heterogeneity, which could be qualitatively different from the geographic distribution.  It should be noted that before COVID-19, several authors had already studied such kind of general heterogeneous models (\cite{Katriel2012}, \cite{Novozhilov2008}, \cite{Novozhilov2012}, \cite{Thieme2009}).
As is mentioned above, the practical importance of the heterogeneity model has been widely discussed by many authors in recent years, but its rigorous mathematical foundation is still unclear.

In this paper, we consider the infection-age-dependent Kermack--McKendrick model, where host individuals are distributed in a continuous state space $\Omega \subset \mathbb R^n$.
To provide mathematical foundation for the heterogeneous model with unbounded susceptibility and non-compact domain, we develop a $L^1$-framework to formulate basic epidemiological concepts.
As is shown in Diekmann and Inaba \cite{Diekmann2023}, the basic model is written as a simple renewal equation, which essentially has only one model component (the force of infection), and can be reduced to a variety of compartmental models under certain assumptions.  First, we prove the well-posedness of the basic system.  Next we define the basic reproduction number and prove the pandemic threshold theorem.  We then present a systematic procedure to compute the effective reproduction number and the herd immunity threshold. Finally we give some illustrative examples and concrete results by using the separable mixing assumption.  In particular, we discuss the critical condition for epidemic resurgence.

\section{The general Kermack--McKendrick model with static heterogeneity}\label{sec2}

We consider a closed population with static individual heterogeneity. Let $S(t,x)$ be the density of completely susceptible, never infected population with individual trait $x$ at time $t$. From our biological interpretation, it is reasonable to assume that $S(t,\cdot) \in L^1_+(\Omega)$. 
The Kermack and McKendrick epidemic model with individual static heterogeneity $x \in \Omega$ for a closed population is formulated as follows \cite{Bootsma2023a, Bootsma2023b, Diekmann2023}:

\begin{equation}\label{sir1}
\begin{aligned}
&\frac{\partial S(t,x)}{\partial t}=-F(t,x) S(t,x),  \cr
&F(t,x)= \int_\Omega \int_{0}^{\infty}A(\tau, x, \sigma)F(t-\tau,\sigma)S(t-\tau,\sigma)d\tau d\sigma, 
\end{aligned}
\end{equation}
where $x \in \Omega$ is the individual heterogeneity (trait) parameter, $\Omega \subset \mathbb R^n$ is its state space and $A(\tau, x,\sigma)$ denotes the expected contribution from infecteds with trait $\sigma$ to the force of infection applied to susceptibles of trait $x$ as a function of the time $\tau$ elapsed since infection took place.

For simplicity, we consider only the case of continuous heterogeneity distribution.  To deal with a mixture of continuous and discrete heterogeneity, we can use the measure-theoretic formulation \cite{Diekmann2023}.  For basic observations about the model \eqref{sir1}, the reader can refer to chapter 6 of \cite{Diekmann2000} or chapter 8 of \cite{Diekmann2013}.  It is also noted that this kind of the renewal equation formulation is also useful to consider the Kermack--McKendrick {\it endemic} model \cite{Breda2012}.  The reader may find McKendrick equation (PDE) approach to the Kermack--McKendrick model in \cite{Inaba2001,Inaba2014, Inaba2016, Inaba2017}.

Under additional assumptions, the basic system \eqref{sir1} can be reduced to a variety of compartment models (see \cite{Diekmann2013}, Chap. 8).  Here a simple example is given to show the reduction of the basic integral system to the traditional compartmental SIR model. 

\subsubsection*{\bf{Example 1}}
Let $\gamma$ be the recovery rate and let $I(t,x)$ be the density of infected individuals.  Then we have
\begin{equation}
I(t,x)=\int_{-\infty}^{t}e^{-\gamma(t-\tau)}F(\tau,x)S(\tau,x)d\tau,
\end{equation}
and it follows that
\begin{equation}
\frac{\partial I(t,x)}{\partial t}=-\gamma I(t,x)+F(t,x)S(t,x),
\end{equation}
which can be supplemented with 
\begin{equation}
\frac{\partial R(t,x)}{\partial t}=\gamma I(t,x),
\end{equation}
where $R(t,x)$ denotes the density of recovered individuals.  
Suppose that $A(\tau,x,\sigma)=e^{-\gamma \tau}\beta(x,\sigma)$.  Then we have
\begin{equation}\label{F}
\begin{aligned}
F(t,x)&=\int_\Omega \int_{0}^{\infty}A(\tau, x, \sigma)F(t-\tau,\sigma)S(t-\tau,\sigma)d\tau d\sigma \cr 
&=\int_{\Omega}\beta(x,\sigma)I(t,\sigma)d\sigma.
\end{aligned}
\end{equation}
Then the Kermack and McKendrick model \eqref{sir1} is reduced to the ``SIR'' (compartmental) epidemic model with trait varibale $x$:
\begin{equation}\label{SIR}
\begin{aligned}
&\frac{\partial S(t,x)}{\partial t}=- S(t,x)\int_{\Omega}\beta(x,\sigma)I(t,\sigma)d\sigma,  \cr
&\frac{\partial I(t,x)}{\partial t}=S(t,x)\int_{\Omega}\beta(x,\sigma)I(t,\sigma)d\sigma-\gamma I(t,x),\cr
&\frac{\partial R(t,x)}{\partial t}=\gamma I(t,x).
\end{aligned}
\end{equation}

Kendall \cite{Kendall1957} used the above SIR model \eqref{SIR} with heterogeneity to study the geographic spread of the epidemic, and formulated the Pandemic Threshold Theorem (PTT), where $x$ is the spatial location of the individual.  
Hadeler \cite{Hadeler2017} considered this type of model by assuming that $x$ is the chronological age of individuals.  In fact, as far as we consider fast epidemics compared to the demographic timescale, aging and demographic turnover have no remarkable effect on the epidemics, so the chronological age of individuals can be considered as a static heterogeneity.  Moreover, as is shown in \cite{Tkachenko2021a}, \cite{Montalban2022} and \cite{Diekmann2023}, the heterogeneous SIR model \eqref{SIR} can be reduced to a ODE system with nonlinear incidence rate if the host heterogeneity distribution is given by the gamma distribution and $\beta(x,\sigma)$ is proportional to $x$.
As was pointed out in Thieme \cite{Thieme2009}, the infinite-dimensional ODE model is mathaematically quite difficult if the parameters are not bounded in the non-compact domain, so we start from the renewal integral equation.

\section{The initial value problem}

For the heterogeneous model \eqref{sir1}, the existence and uniqueness for the total orbit starting from the disease-free steady state $S(-\infty,x)=N(x)$ is still an open problem, although it has been solved for the homogeneous case \cite{Diekmann1977}. Instead, following Rass and Radcliff \cite{Rass2003}, we consider here the initial value problem for the basic system \eqref{sir1}.  

Let $N(\cdot) \in X_+:=L^1_+(\Omega)$ be the distribution of host population heterogeneity (traits) in a steady state and $\|N\|_X$ is the total host size: $\|N\|_X:=\int_\Omega \abs{N(x)} dx$, where $\|\cdot \|_X$ denotes the $L^1$ norm of the space $X:=L^1(\Omega)$. 
We assume that at time $t=0$, the totally susceptible host population $N(x)$ is exposed to infection from the outside.  
So $S(0,x)=N(x)$ and it holds that

\begin{equation}\label{sir3}
\begin{aligned}
&\frac{\partial S(t,x)}{\partial t}=-F(t,x)S(t,x),  \cr
&F(t,x)= \int_\Omega \int_{0}^{t}A(\tau, x, \sigma) F(t-\tau,\sigma)S(t-\tau,\sigma) d\tau d\sigma+G(t,x),
\end{aligned}
\end{equation}
where $t>0$ and $G(t,x)$ is the given initial data for the force of infection, which is assumed to be generated by an external infection.

For the initial value problem, its mathematical well-posedness can be established by reducing it to the problem of the integral equation for the cumulative force of infection.
Let $W(t,x)$ be the cumulative force of infection:
\begin{equation}\label{3.2}
W(t,x):=\int_{0}^{t}F(\tau, x)d\tau.
\end{equation}
Then we have 
\begin{equation}\label{3.3}
S(t,x)=N(x)e^{-W(t,x)},
\end{equation}
 and
\begin{equation}\label{sir*}
\begin{aligned}
W(t,x)&=H(t,x)+\int_{0}^{t}d\eta \int_\Omega \int_{0}^{\eta}A(\tau, x, \sigma) \left(-\frac{\partial S}{\partial \eta}(\eta-\tau,\sigma)\right) d\tau d\sigma \cr
&=H(t,x)+ \int_\Omega \int_{0}^{t}A(\tau, x, \sigma)(N(\sigma)-S(t-\tau,\sigma))d\tau d\sigma \cr
&=H(t,x)+ \int_\Omega \int_{0}^{t}A(\tau, x, \sigma)N(\sigma)(1-e^{-W(t-\tau,\sigma)})d\tau d\sigma \cr
\end{aligned}
\end{equation}
where 
\begin{equation}
H(t,x):=\int_{0}^{t}G(\tau,x)d\tau.
\end{equation}

The integral equation \eqref{sir*} was traditionally studied as the Volterra-Hammerstein integral equation \cite{Thieme1977a, Thieme1977b, Thieme1980}. Diekmann \cite{Diekmann1978} used the renewal equation \eqref{sir*} to analyze the Kermack--McKendrick model with spatial heterogeneity. 
For traditional space-dependent models, much attention has been paid to the spatial interaction kernel like as $A(\tau,x,\sigma)=u(\tau)v(x-\sigma)$ and its traveling wave solutions in a noncompact domain $\Omega$.  On the other hand, recent COVID-19 inspired studies assumes that the feature is not necessarily the spatial heterogeneity, but the biological or social individual heterogeneity, so the separable kernel as $A(\tau,x,\sigma)=a(x)b(\tau)c(\sigma)$ played a much more important rule \cite{Diekmann2023}.   

To treat unbounded susceptibility and its non-compact state space, we first give  existence and uniqueness results for the solution of \eqref{sir*} under new assumptions different from \cite{Diekmann1978} and Inaba \cite{Inaba2014}. 
Since $F(t,x)S(t,x)$ should be integrable with respect to $x$, it is reasonable to expect that $W \in BC_+(\mathbb R_+;Y)$, where $Y$ be a Banach space defined by $Y:=\left\{N^{-1}\psi: \psi \in  X\right\}$ with norm $\|\phi\|_Y=\int_{\Omega} \abs{\phi(x)} N(x)dx$. 
Then we have  $\|\phi\|_Y=\|\phi N\|_X$.  
Hence, we assume that a map $P: \phi \to N\phi$ is an isometry from a Banach space $Y$ to a Banach space $X$. 
For any $T>0$, let $C_T:=C([0,T];Y)$ be the set of continuous functions on $[0,T]$ equipped with the norm $\|f\|_{C_T}:=\sup_{0 \le t \le T}\|f(t,\cdot)\|_Y$.  
We also make the following technical assumptions: 

\begin{assumption}\label{ass1}
\begin{enumerate}
\item There exist nonnegative measurable functions $a$, $b$, $c$ and a number $\alpha>1$ such that
\begin{equation}\label{ass11}
a(x)b(\tau)c(\sigma) \le A(\tau,x,\sigma) \le \alpha a(x)b(\tau)c(\sigma).
\end{equation}
where  $a \in Y_+$, $b \in L^1_+(\mathbb R_+)$ and $c^n \in Y_+$ for $n=1, 2$.
Furthermore, there exists a bounded continuous function $Q \in BC_+(\mathbb R_+)$ such that
\begin{equation}\label{ass111}
a(x)Q(t) \le H(t,x) \le \alpha a(x)Q(t).
\end{equation}
\item It holds that
\begin{equation}\label{ass12}
\lim_{h \to 0}\int_\Omega \int_\Omega  \int_{0}^{\infty}N(x) \abs{A(\tau+h,x,\sigma)-A(\tau,x,\sigma)} N(\sigma)d\tau d\sigma dx=0,
\end{equation}
where we assume that $A(\tau+h,x,\sigma)=0$ if $\tau+h \notin \mathbb R_+$ and

\begin{equation}\label{ass13}
\lim_{h \to 0}\int_\Omega \int_\Omega  \int_{0}^{\infty} \abs{N(x+h)A(\tau,x+h,\sigma)-N(x)A(\tau,x,\sigma)} N(\sigma)d\tau d\sigma dx=0,
\end{equation}
where we assume that $N(x+h)A(\tau,x+h,\sigma)=0$ if $x+h \notin \Omega$.

\end{enumerate}
\end{assumption}

Let $C_{T+}$ be the positive cone of $C_T$. Define an operator $\mathcal F$ for $\phi \in C_{T+}$ by
\begin{equation}
\mathcal F (\phi)(t,x)=H(t,x)+ \int_\Omega \int_{0}^{t}A(\tau, x, \sigma)N(\sigma)(1-e^{-\phi(t-\tau,\sigma)})d\tau d\sigma,
\end{equation}
where we assume $H \in C_{T+}$.

\begin{proposition}\label{prop1}
 $\mathcal F$ is a continuous map from $C_{T+}$ into itself.
\end{proposition}
   \begin{proof}
   For $\phi \in C_{T+}$, we have
   \begin{equation}\label{3.10}
   \mathcal F(\phi)(t,x) \le H(x,t)+\int_\Omega \int_{0}^{t}A(\tau, x, \sigma)N(\sigma)d\tau d\sigma \le \alpha a(x)Q(t)+\alpha L a(x)\|c\|_Y,
   \end{equation}
where $L:=\int_{0}^{\infty}b(\tau)d\tau$.
   Then for every $t \in [0,T]$, we have $\mathcal F(\phi)(t,\cdot) \in Y_+$ and $\|F(\phi)(t,\cdot)\|_Y$ is uniformly bounded. Let
\[J(t,x;\phi):=\int_\Omega \int_{0}^{t}A(t-\tau, x, \sigma)N(\sigma)(1-e^{-\phi(\tau,\sigma)})d\tau d\sigma.
\]
Then we have
\begin{equation}
\int_\Omega \abs{J(t_1,x;\phi)-J(t_2,x;\phi)} N(x)dx \le Q_1(t_1,t_2)+Q_2(t_1,t_2),
\end{equation}
where
\[Q_1(t_1,t_2):=\int_\Omega dx \int_\Omega d\sigma \int_{0}^{t_1} N(x)\abs{A(t_1-\tau,x,\sigma)-A(t_2-\tau,x,\sigma)} N(\sigma)d\tau,\]
\[Q_2(t_1,t_2):=\int_\Omega dx \int_\Omega d\sigma \int_{t_1}^{t_2} N(x)\abs{A(t_2-\tau,x,\sigma)}N(\sigma) d\tau. \]
From our assumption \eqref{ass11}-\eqref{ass12}, for any $\epsilon>0$ there exists $\delta>0$ such that 
\begin{equation}\label{3.13}
 \|\mathcal F(\phi)(t_1,\cdot)-\mathcal F(\phi)(t_2,\cdot)\|_Y \le  \|H(t_1,\cdot)-H(t_2,\cdot)\|_Y+Q_1(t_1,t_2)+Q_2(t_1,t_2)<\epsilon,
 \end{equation}
 if $\abs{t_1-t_2}<\delta$.  Then $\mathcal F(\phi) \in C_{T+}$.  Next observe that for $\phi_j \in C_{T+}$,
 \begin{equation}
 \|\mathcal F(\phi_1)(t,\cdot)-\mathcal F(\phi_2)(t,\cdot)\|_Y
 \le \alpha L \|a\|_Y  \int_\Omega c(\sigma)N(\sigma) \abs{e^{-\phi_1(t-\tau,\sigma)}-e^{-\phi_2(t-\tau,\sigma)}} d\sigma.
 \end{equation}
 From the Schwarz inequality, it follows that
 \[\begin{aligned}
  &\int_\Omega c(\sigma)N(\sigma) \abs{e^{-\phi_1(t-\tau,\sigma)}-e^{-\phi_2(t-\tau,\sigma)}} d\sigma \cr
&\le  \left(\int_\Omega c^2(\sigma)N(\sigma) \abs{e^{-\phi_1(t-\tau,\sigma)}-e^{-\phi_2(t-\tau,\sigma)}} d\sigma \right)^{\frac{1}{2}} \cr
& ~~~~~~\times \left( \int_\Omega N(\sigma) \abs{e^{-\phi_1(t-\tau,\sigma)}-e^{-\phi_2(t-\tau,\sigma)}} d\sigma \right)^{\frac{1}{2}},
\end{aligned}
\]
where
\[\int_\Omega c^2(\sigma)N(\sigma) \abs{e^{-\phi_1(t-\tau,\sigma)}-e^{-\phi_2(t-\tau,\sigma)}} d\sigma\le 2\|c^2\|_Y,\]
and
\[\begin{aligned}
&\int_\Omega N(\sigma) \abs{e^{-\phi_1(t-\tau,\sigma)}-e^{-\phi_2(t-\tau,\sigma)}} d\sigma \cr & \le\int_\Omega N(\sigma) \abs{\phi_1(t-\tau,\sigma)-\phi_2(t-\tau,\sigma)} d\sigma
\le \|\phi_1-\phi_2\|_{C_T}.
\end{aligned}
\]
 So we have
 \begin{equation}
 \|\mathcal F(\phi_1)-\mathcal F(\phi_2)\|_{C_T} \le \alpha L \|a\|_Y \sqrt{2\|c_2\|_Y\|\phi_1-\phi_2\|_{C_T}}.
 \end{equation}
Then $\mathcal F$ is a continuous positive operator. 
 \end{proof}

\begin{proposition}\label{prop2}
$\mathcal F$ has at least one fixed point. 
\end{proposition}
\begin{proof}
Define a subset $\Phi:=\{\phi_n: n=0,1,2,\cdots\} \subset C_{T+}$, where
\begin{equation}
\phi_0=H, \quad \phi_n=\mathcal F(\phi_{n-1}), ~(n \ge 1).
\end{equation} 
From \eqref{3.13} in the proof of Proposition \ref{prop1}, $\Phi$ is equicontinuous on $[0,T]$.  Next, for each $t \in [0,T]$, define a set $\Phi(t):=\{\phi_n(t,\cdot): n=0,1,2,\cdots\} \subset Y_+$.  From the assumption \ref{ass1}, it follows that $\sup_{n=0,1,2,..}\|\phi_n(t,\cdot)\|_Y<\infty$, and it follows from \eqref{3.10} that 
\begin{equation}
\lim_{r \to \infty}\int_{r}^{\infty} \abs{\phi_n(t,x)} N(x)dx=0,\quad \lim_{h \to +0}\int_{0}^{h} \abs{ \phi_n(t,x)} N(x)dx=0,
\end{equation}
uniformly in $\phi_n(t,\cdot) \in \Phi(t)$.  From \eqref{ass13}, we have
\begin{equation}
\lim_{h \to 0}\int_\Omega \abs{ N(x+h)\phi_n(t,x+h)-N(x)\phi_n(t,x)} dx=0,
\end{equation}
uniformly in $\phi_n(t,\cdot) \in \Phi(t)$.  
Then it follows from the Fr\'echet-Kolmogorov criterion (Theorem B.2. in \cite{Smith2011}) 
that the set $\{\phi_n(t,\cdot)N\}_{n=1,2,...} \subset L^1_+(\Omega)$ has compact closure.  
Then we can choose a convergent subsequence from $\Phi(t)$, so $\Phi(t)$ is relatively compact in $Y_+$.  Thanks to Ascoli's theorem \cite{Lang1993}, $\Phi$ is also relatively compact in $C_{T+}$.  Then we can choose a convergent subsequence $\{\phi_{n(k)}\}_{k=1,2,...} \subset \Phi$,  and we conclude that $\lim_{k \to \infty}\phi_{n(k)}=\phi_\infty \in C_{T+}$ exists and $\phi_\infty=\lim_{k \to \infty} \mathcal F(\phi_{n(k-1)})=\mathcal F(\lim_{k \to \infty} \phi_{n(k-1)})=\mathcal F(\phi_\infty)$. Thus $\phi_\infty$ is a positive fixed point of $\mathcal F$. 
\end{proof}

\begin{proposition}\label{prop3}
 $\mathcal F$ has at most one positive fixed point.
\end{proposition}
\begin{proof}
At first, note that the operator $\mathcal F$ is monotonically non-decreasing positive operator. From our assumption \ref{ass1}, we have
\begin{equation}
h(\phi)(t)a(x) \le \mathcal F(\phi)(t,x)\le \alpha h(\phi)(t) a(x),
\end{equation}
where $h: C_{T+} \to BC_+([0,T]:\mathbb R_+)$ is given by
\begin{equation}
h(\phi)(t):=Q(t)+ \int_\Omega \int_{0}^{t}b(\tau)c(\sigma)N(\sigma)(1-e^{-\phi(t-\tau,\sigma)})d\tau d\sigma.
\end{equation}
It also holds that for $s \in (0,1)$,
\begin{equation}
\mathcal F(s\phi) (t,x) \ge s\mathcal F(\phi) (t,x) + \eta(\phi,s)(t)a(x),
\end{equation}
where
\begin{equation}\label{eta}
\eta(\phi,s)(t):= \int_\Omega \int_{0}^{t}b(\tau)c(\sigma)N(\sigma)(1-e^{-s \phi(t-\tau,\sigma)}-s(1-e^{ -\phi(t-\tau,\sigma)})d\tau d\sigma.
\end{equation}
Then it is easy to see that $\eta(\phi,s)(t)>0$ for any $\phi \in C_{T+} \setminus\{0\}$ and any $s \in (0,1)$.
Suppose that $\mathcal F$ has two non-zero fixed points $\phi_1$ and $\phi_2$.  Then for every $t>0$, $h(\phi_j)(t,x)>0$ and it follows that
\begin{equation}\label{unique1}
\phi_1(t,x)=\mathcal F(\phi_1)(t,x) \ge h(\phi_1)(t)a(x) \ge \frac{h(\phi_1)(t)}{\alpha h(\phi_2)(t)}\mathcal F(\phi_2)(t,x)=\frac{h(\phi_1)(t)}{\alpha h(\phi_2)(t)}\phi_2(t,x).
\end{equation}
For a fixed $t>0$, define a number $k:=\sup \{\mu: \phi_1(t,x)  \ge \mu \phi_2(t,x)\}$. It follows from \eqref{unique1} that $k>0$. Suppose that $0<k<1$.  Then we can observe that
\begin{equation}\label{unique2}
\begin{aligned}
\phi_1(t,x)&=\mathcal F(\phi_1)(t,x) \ge \mathcal F(k\phi_2)(t,x) \ge k\mathcal F(\phi_2) (t,x) + \eta(\phi_2,k)(t)a(x) \cr
&=k\phi_2 (t,x) + \eta(\phi_2,k)(t)a(x) \ge \left(k+\frac{\eta(\phi_2,k)}{\alpha h(\phi_2)(t)}\right)\phi_2(t,x),
\end{aligned}
\end{equation}
which contradicts the definition of $k$.  Then we conclude that $k \ge 1$ and $\phi_1(t,x) \ge \phi_2(t,x)$.  Switching the roles of $\phi_1$ and $\phi_2$, we can repeat the same argument to prove that $\phi_2(t,x) \ge \phi_1(t,x)$, so we conclude that $\phi_1=\phi_2$ for every $t \in [0,T]$.  Then $\mathcal F$ has at most one positive fixed point
\end{proof}

From Propositions \ref{prop2} and \ref{prop3}, the equation \eqref{sir*} has a unique positive solution for any $T>0$, and it is uniformly bounded, so the global positive solution of \eqref{sir*} exists uniquely for $t \in \mathbb R_+$.

\begin{proposition}\label{prop4}
For the global solution of \eqref{sir*}, there exists $W(\infty,\cdot) \in Y_+$ such that $\lim_{t \to \infty}\|W(t,\cdot)-W(\infty,\cdot)\|_Y=0$ and $W(\infty,x)$ satisfies the limit equation
\begin{equation}\label{sir**}
W(\infty,x)=H(\infty,x)+\int_\Omega \Theta(x, \sigma)N(\sigma)(1-e^{-W(\infty,\sigma)}) d\sigma,
\end{equation}
where
\begin{equation}
\Theta(x,\sigma):=\int_{0}^{\infty}A(\tau,x,\sigma)d\tau.
\end{equation}
\end{proposition}
\begin{proof}
Since $W(t,\cdot)N \in X=L^1_+(\Omega)$ is monotonically non-decreasing with respect to $t$ and $\sup_{t>0}\|W(t,\cdot)\|_Y<\infty$, it follows from B. Levi's theorem that there exists a $M \in X$ such that $\lim_{t \to \infty}W(t,\cdot)N=M \in X$.  Then $\lim_{t \to \infty}\|W(t,\cdot)-W(\infty,\cdot)\|_Y=0$ where $W(\infty,\cdot):=M/N$.   Observe that
 \[
\begin{aligned}
&\abs[\Big]{ \int_\Omega d\sigma \int_{0}^{t}A(\tau,x,\sigma)N(\sigma)(1-e^{-W(t-\tau,\sigma)})d\tau-\int_\Omega d\sigma \int_{0}^{\infty}A(\tau,x,\sigma)N(\sigma)(1-e^{-W(\infty,\sigma)})d\tau
 } \cr
&\le  \alpha a(x)\int_{0}^{t}b(\tau) d\tau \int_\Omega c(\sigma)N(\sigma) \abs{ e^{-W(t-\tau,\sigma)}-e^{-W(\infty,\sigma)} } d\sigma \cr
&~~~+ \alpha a(x)\int_{t}^{\infty}b(\tau) d\tau \int_\Omega c(\sigma)N(\sigma)e^{-W(\infty,\sigma)}d\sigma,
\end{aligned}
\]
 where it is clear that the second term on the right-hand side goes to zero when $t \to \infty$.  For the first term, we observe that
\[
\begin{aligned}
&\int_{0}^{t}d\tau b(\tau)  \int_\Omega c(\sigma)N(\sigma) \abs{ e^{-W(t-\tau,\sigma)}-e^{-W(\infty,\sigma)}} d\sigma \cr
&\le \int_{0}^{t}b(\tau) d\tau \int_\Omega c(\sigma)N(\sigma) \abs{ W(t-\tau,\sigma)-W(\infty,\sigma)} d\sigma \cr
&\le \int_{0}^{t}d\tau b(\tau)\left(2\|c^2N\|_X\|W(t-\tau,\cdot)-W(\infty,\cdot)\|_Y \right)^\frac{1}{2},
\end{aligned}
\]
where we use the Schwarz inequality as in the proof of Proposition 1.  Since
\[\begin{aligned}
&b(\tau)\left(2\|c^2N\|_X\|W(t-\tau,\cdot)-W(\infty,\cdot)\|_Y\right)^{\frac{1}{2}} \cr
&~~\le b(\tau)\left(2\|c^2N\|_X\sup_{t \ge 0}\|W(t,\sigma)-W(\infty,\sigma)\|_Y \right)^{\frac{1}{2}} \in L^1(\mathbb R_+),
\end{aligned}
\]
then it follows from the dominated convergence theorem that 
\[\lim_{t \to \infty}\int_{0}^{t}d\tau b(\tau)\left(2\|c^2N\|_X\|W(t-\tau,\cdot)-W(\infty,\cdot)\|_Y \right)^\frac{1}{2}=0,
\]
Si if we let $t \to \infty$ in the equation \eqref{sir*}, we have \eqref{sir**}.

\end{proof}

\section{The basic reproduction number}

Let $B(t,x):=F(t,x)S(t,x)$ be the incidence at time $t>0$.  Then we can rewrite \eqref{sir3} as an initial value problem for $S$ and $B$:
\begin{equation}\label{sir4}
\begin{aligned}
&\frac{\partial S(t,x)}{\partial t}=-B(t,x),  \cr
&B(t,x)= S(t,x)\left[\int_\Omega \int_{0}^{t}A(\tau, x, \sigma) B(t-\tau, \sigma) d\tau d\sigma+G(t,x)\right]. 
\end{aligned}
\end{equation}
In the disease invasion phase, the incidence is described by the linearized equation as
\begin{equation}\label{LRE}
B(t,x)= N(x)\int_\Omega \int_{0}^{t}A(\tau, x, \sigma) B(t-\tau, \sigma) d\tau d\sigma+N(x)G(t,x).
\end{equation}
where $B$ denotes the incidence rate in the invasion phase for the totally susceptible population.
Define the net reproduction operator acting on the birth state space $X=L^1_+(\Omega)$ as
\begin{equation}
(K(\tau)\phi)(x):=N(x)\int_\Omega A(\tau, x, \sigma)\phi(\sigma)d\sigma, ~\phi \in X.
\end{equation}
Then \eqref{LRE} is formulated as an abstract renewal equation
\begin{equation}\label{aRE}
B(t)=NG(t)+\int_{0}^{t}K(\tau)B(t-\tau)d\tau,
\end{equation}
where $B(t)=B(t,\cdot)$ and $NG(t)=N(\cdot)G(t,\cdot)$ are vector-valued functions from $\mathbb R_+$ to $X_+$. 
 
The next generation operator (NGO)\footnote{The NGO was first introduced in \cite{Diekmann1990},  see also \cite{Inaba2017}.} acting on $X$ associated with the renewal process \eqref{aRE} is defined by
\begin{equation}
T:=\int_{0}^{\infty}K(\tau)d\tau,
\end{equation}
which is the linear integral operator defined by
\begin{equation}
(T\phi)(x)=N(x)\int_{\Omega} \Theta(x,\sigma)\phi(\sigma)d\sigma, ~\phi \in X.
\end{equation}

Now we introduce the additional condition:

\begin{assumption}\label{ass2}
\begin{enumerate}
\item
For any $f \in L^\infty_+ \setminus \{0\}$, it follows that
\begin{equation}
\langle f, aN \rangle=\int_{\Omega}f(x)a(x)N(x)dx >0.
\end{equation}
\item
 $c \in L^\infty_+(\Omega)$ and 
\begin{equation}
\langle c, \phi \rangle=\int_{\Omega}c(x)\phi(x)dx >0,
\end{equation}
 for all $\phi \in X_+\setminus\{0\}$.  
 \item
 The following holds uniformly for $\sigma \in \Omega$,
\begin{equation}\label{FK}
\lim_{h \to 0}\int_{\Omega} \abs{ N(x+h)\Theta(x+h,\sigma)-N(x)\Theta(x,\sigma) } dx=0,
\end{equation}
where we assume that $\Theta(x,\sigma)=0$ for $x \notin \Omega$.
\end{enumerate}
\end{assumption}

Under the assumptions \ref{ass1}-\ref{ass2}, the next generation operator $T$ is a positive, bounded linear operator, and we can show that

\begin{proposition}\label{prop5}
Under the assumptions \ref{ass1}-\ref{ass2}, $T$ is compact and nonsupporting, so $r(T)$ is the dominant eigenvalue of $T$.
\end{proposition}
\begin{proof}
From 
\begin{equation}
   L a(x)N(x)\langle c,\phi \rangle \le  (T\phi)(x) \le \alpha La(x)N(x)\langle c,\phi \rangle.
\end{equation}
 We know that $T$ is a bounded linear operator on $X$ and it follows that for $\phi \in X_+$,
and for every integer $n \ge 1$, it holds that
 \begin{equation}
 (T^n\phi)(x) \ge L^n \langle c, \phi \rangle \langle c, Na\rangle^{n-1}N(x)a(x).
 \end{equation}
 By the assumption \ref{ass2},  $T^n\phi$ is a quasi-interior point in the cone $X_+$.  Then $T$ is a strictly nonsupporting operator (Definition 10.3, \cite{Inaba2017}).  It follows from the positive operator theory that the spectral radius $r(T)$ is the dominant eigenvalue of $T$.  
Next we show the compactness of $T$ for the case that $\Omega=[0,\infty)$.  Let $\phi \in U:=\{ \phi \in X: \|\phi\|_X \le M\}$.  Then we have $\|T\phi \|_X \le \alpha L\|c\|_\infty \|a\|_Y M$, where $\|c\|_\infty:=\sup_{x \in \Omega}\abs{c(x)}$.  Next we can see that
\[ \lim_{r \to \infty}\int_{r}^\infty \abs{ (T\phi)(x) } dx \le \lim_{r \to \infty}\alpha LM \|c\|_\infty \int_{r}^{\infty}N(x)a(x)dx = 0,\]
\[\lim_{h \to +0}\int_{0}^{h} \abs{ (T\phi)(x) } dx \le \lim_{h \to +0}\alpha M L\|c\|_\infty \int_{0}^{h}N(x)a(x)dx = 0,\]
where the convergence is uniform in $\phi \in U$. Finally it follows from the condition \eqref{FK} that
\[\lim_{h \to +0}\int_{0}^{\infty} \abs{ (T\phi)(x+h)-(T\phi)(x) } dx =0,\]
uniformly in $\phi \in U$.
Then we can use the Fr\'echet-Kolomogorov criterion (see \cite{Smith2011}, Theorem B.2.) to conclude that $T(U)$ has compact closure.
\end{proof}

Under the assumptions \ref{ass1}-\ref{ass2}, the basic reproduction number $R_0$ is defined by its spectral radius $r(T)$.  In fact, it follows from the well-known Renewal Theorem that  $B(t)$ is asymptotically proportional to $e^{r_0 t}\phi(x)$, where the growth rate $r_0 \in \mathbb R$ and the density $\phi \in X_+$ satisfy the eigenvalue problem $\phi=\hat{K}(r_0) \phi$,
where
 $\hat{K}(\lambda)$, $\lambda \in \mathbb C$ is the Laplace transform of $K$ defined by
$\hat{K}(\lambda):=\int_{0}^{\infty}e^{-\lambda \tau}K(\tau)d\tau$.
Then the spectral radius $r(\hat{K}(\lambda))$, $\lambda \in \mathbb R$ is the positive eigenvalue of $\hat{K}(\lambda)$. 
Then the intrinsic growth rate (asymptotic Malthusian parameter for $B$) $r_0$, is given as the real root such that $r(\hat{K}(r_0))=1$ and $\phi$ is the positive eigenvector of $\hat{K}(r_0)$ associated with its eigenvalue unity.
Since the spectral radius $r(\hat{K}(\lambda))$, $\lambda \in \mathbb R$ is monotonically decreasing with respect to $\lambda$, so the sign relation ${\rm sign}(r_0)={\rm sign}(R_0-1)$ holds,
where $R_0=r(\hat{K}(0))=r(T)$. (\cite{Inaba2017}, chapter 10).

\section{Pandemic Threshold Theorem}

As is mentioned in Section 3, the generalized pandemic threshold theorem has been proved by Thieme \cite{Thieme1977a}, Diekmann \cite{Diekmann1978}, and Inaba \cite{Inaba2014} under different formulations, respectively.  Here we give a simple proof under the new assumption. 

From \eqref{3.3}, $\lim_{t \to \infty}S(t,x)$ exists and it holds that
\begin{equation}
W(\infty,x)=-\log \left(\frac{S_\infty(x)}{N(x)} \right).
\end{equation}
Let $t \to \infty$ in \eqref{sir*}, we have
\begin{equation}\label{20}
W(\infty,x)=H(\infty,x)+\Psi(W(\infty,\cdot))(x).
\end{equation}
where $\Psi$ is 
 a nonlinear operator from $Y_+$ to $Y_+$ defined by
\begin{equation}
\begin{aligned}
(\Psi \psi)(x)=\int_{\Omega}\Theta(x,\sigma)N(\sigma)(1-e^{-\psi(\sigma)})d\sigma, \quad \psi \in Y_+
\end{aligned}
\end{equation}
where $H(\infty,\cdot) \in Y_+$.  
To estimate $W(\infty,\cdot)$, the fixed point of $\Psi$ plays a key role.  

\begin{proposition}\label{prop6}
If $R_0 \le 1$, $\Psi$ has no positive fixed point.
\end{proposition}
\begin{proof}
Let $\psi \in Y_+$ be a positive fixed point in $Y_+$.  Then we have $N \psi \in X_+$ and
\begin{equation}\label{47}
\begin{aligned}
N(x)\psi(x) &=N(x)\int_{\Omega}\Theta(x,\sigma)N(\sigma)(1-e^{-\psi(\sigma)})d\sigma \cr
& \le N(x)\int_{\Omega}\Theta(x,\sigma)N(\sigma)\psi(\sigma)d\sigma=(T N\psi)(x),
\end{aligned}
\end{equation}
Let $f^* \in X_+^*$ be the adjoint eigenfunctional of $T$ associated with the eigenvalue $r(T)=R_0$.
Taking the duality pairing in \eqref{47}, we have
\[ \langle f^*, N\psi \rangle \le \langle f^*, T N\psi \rangle=R_0 \langle f^*, N\psi \rangle,\]
where $\langle f^*,\phi \rangle$ denotes the value of $f^*$ at $\phi \in X$, and the equality holds only if $\psi=0$.  Then we have $R_0>1$.  So if $R_0 \le 1$, there is no positive fixed point.
\end{proof}

Although we skip the proof, it is intuitively clear that $W(\infty,\cdot)$ is monotonically non-decreasing with respect to $H(\infty,\cdot)$.  Then Proposition \ref{prop6} implies that if $R_0 \le 1$, $W(\infty,\cdot) \to 0$ when $H(\infty,\cdot) \to 0$.

Let $P: Y \to X$ be the isomorphism between $X$ and $Y$ such that $P\psi=N\psi$.
Define an operator $\Phi$ from $X$ into itself by $\Phi=P\Psi P^{-1}$.  Then we have
\begin{equation}
\Phi(\phi)(x)=N(x)\int_{\Omega}\Theta(x,\sigma)N(\sigma)(1-e^{-N^{-1}(\sigma)\phi(\sigma)})d\sigma, \quad \phi \in X_+.
\end{equation}

\begin{proposition}\label{prop7}
Under the assumption \ref{ass1}-\ref{ass2},  $\Phi$ is a compact and concave operator\footnote{The definition of the concave operator is given in Krasnoselskii \cite{Krasnoselskii1964}. The reader may consult \cite{Inaba1990} and \cite{Inaba2017}.} in $X$.
\end{proposition}
\begin{proof}
First we show that $\Phi$ is a compact operator.  For simplicity, we consider the case that $\Omega=[0,\infty)$. As is seen in Proposition \ref{prop5}, it is sufficient to check the conditions for the compactness of $\Phi(U) \subset X=L^1(\Omega)$ for any bounded set $U \subset X_+$ (Theorem B.2. in \cite{Smith2011}).  
In fact, it is easy to see $\|\Phi(\phi) \|_X \le \alpha L \|Na\|_X \langle c, N \rangle$ for any $\phi \in X_+$. Then the condition $\sup_{\phi \in X_+}\|\Phi(\phi) \|_X <\infty$ holds.  Next it holds that
\[ \lim_{r \to \infty}\int_{r}^\infty \abs{ \Phi(\phi)(x) } dx \le \lim_{r \to \infty}\alpha L \langle c, N \rangle \int_{r}^{\infty}N(x)a(x)dx=0,\]
\[ \lim_{h \to +0}\int_{0}^{h} \abs{ \Phi(\phi)(x) } dx \le \lim_{r \to \infty}\alpha L \langle c, N \rangle\int_{0}^{h}N(x)a(x)dx=0.\]
Finally let
\[J:=\int_{0}^{\infty} \abs{ \Phi(\phi)(x+h)-\Phi(\phi)(x) } dx.\]
Then we can observe that for $\phi \in X_+$,
\begin{equation}\label{5.6}
\int_\Omega \abs{ \Phi(\phi)(x+h)-\Phi(\phi)(x)} dx \le \int_{\Omega} \int_{\Omega} \abs{ N(x+h)\Theta(x+h,\sigma)-N(x)\Theta(x,\sigma) }  dx N(\sigma)d\sigma. 
\end{equation}
where it follows from the assumption \ref{ass2} that the right-hand side of \eqref{5.6} goes to zero uniformly for $\phi \in X_+$ when $h \to 0$.  Then we conclude that $\Phi(X_+)$ has compact closure. 
Next, we prove the concavity of $\Phi$. 
First it is clear that $\Phi$ is monotone non-decreasing.
Next, from the fact that $1-e^{-tx} \ge t(1-e^{-x})$ for $0\le t \le 1$ and $x \ge 0$, it follows that for $0 \le t \le 1$ and $\phi \in X_+$,
\begin{equation}\label{cond1}
\Phi(t\phi) \ge t\Phi(\phi), \quad 0 \le t \le 1,
\end{equation}
Now we can observe that $\Phi$ is {\it $u_0$-positive}, that is, there exists a nonzero element $u_0 \in X_+$ such that for a nonzero $\phi \in X_+$ the inequality
\begin{equation}\label{cond2}
\epsilon(\phi) u_0 \le \Phi(\phi) \le \alpha \epsilon(\phi) u_0,
\end{equation}
holds, where $\epsilon$ denotes a positive functional on $X_+$.
In fact, under the assumption \ref{ass1}, it holds that for $\phi \in X_+$,
\begin{equation}
La(x) N(x)\epsilon(\phi)  \le \Phi(\phi)(x) \le \alpha L a(x)N(x)\epsilon(\phi),
\end{equation}
where $\epsilon$ is a strictly positive functional on $X_+$ given by 
\begin{equation}
\epsilon(\phi):=\int_{\Omega}c(\sigma)N(\sigma)(1-e^{-N^{-1}(\sigma) \phi(\sigma)})d\sigma.
\end{equation}
Therefore if we choose $u_0(x)=L a(x)N(x)$, we obtain \eqref{cond2}, and \eqref{cond1} and \eqref{cond2} are sufficient conditions for concavity of $\Phi$.
\end{proof}

\begin{proposition}\label{prop8}
$\Phi$ has at most one positive fixed point.
\end{proposition}
\begin{proof}
Observe that for $k \in (0,1)$,
\begin{equation}
\begin{aligned}
\Phi(k\phi)-k\Phi(\phi)&=
N(x)\int_{\Omega}\Theta(x,\sigma)N(\sigma)
(1-e^{-k N^{-1}(\sigma)}-k(1-e^{ -N^{-1}(\sigma)\phi(\sigma)}))d\sigma \cr
&\ge a(x)N(x) \eta(\phi;k),
\end{aligned}
\end{equation}
where
\[\eta(\phi;k):=L\int_{\Omega}c(\sigma)N(\sigma)
(1-e^{-k N^{-1}(\sigma)\phi(\sigma)}-k(1-e^{ -N^{-1}(\sigma)\phi(\sigma)}))d\sigma.\]
Then it is easy to see that $\eta(\phi,k)>0$ for any $\phi \in X_+\setminus\{0\}$ and $k \in (0,1)$.
Therefore, $\Phi$ is monotone and concave in $X_+$ and for any $k \in (0,1)$, it follows that 
\begin{equation}\label{epositive}
\Phi(k\phi) \ge k\Phi(\phi)+\eta(\phi,k) a(x)N(x).
\end{equation}
where $aN \in X_+$ is a quasi-interior point. Then we can apply Lemma 4.8 in \cite{Inaba1990}, $\Phi$ has at most one positive fixed point.
\end{proof}

\begin{proposition}\label{prop9}
If $R_0 > 1$, $\Psi$ has a unique positive fixed point.
\end{proposition}
\begin{proof}
It is easy to see that $\Phi$ maps the positive cone $X_+$ into the bounded convex subset $\{\phi \in X_+: \|\phi\|_X \le \alpha L \|Na\|_X\}$, and its Fr\'echet derivative at the origin, denoted by $\Phi'[0]$, is the next generation operator $T$.  
Then it follows from Krasnoselskii's theorem (Theorem 4.11 in \cite{Krasnoselskii1964}, \cite{Inaba1990}, \cite{Inaba2017}) that $\Phi$ has at least one positive fixed point if $r(\Phi'[0])=r(T)=R_0>1$, and it is unique from Proposition \ref{prop8}.  Therefore, 
$\Psi=P^{-1}\Phi P$ has a unique fixed point in $Y_+$ 
\end{proof}

\begin{proposition}\label{prop10}
Suppose that $\Psi$ has a unique positive fixed point $V \in Y_+ \setminus\{0\}$. Then it holds that
\begin{equation}
W(\infty,x) \ge V(x).
\end{equation}
\end{proposition}
\begin{proof}
Define a positive operator $\mathcal H$ in $Y_+$ by $(\mathcal H\phi)(x):=H(\infty,x)+\Psi(\phi)(x)$.  Then $W(\infty,\cdot)$ is its fixed point and $\mathcal H$ is a monotone operator.  Define a sequence $\phi_n$ by $\phi_n=\mathcal H(\phi_{n-1})$ ($n \ge 1$) and $\phi_0=V$.  Then we have $\phi_1 \ge \Psi(\phi_0)=\phi_0$, so it follows iteratively that $\phi_n \ge \phi_{n-1}$.  Then we conclude that
$\lim_{n \to \infty}\phi_n(x)=W(\infty,x) \ge \phi_0(x)=V(x)$.
\end{proof}

From propositions \ref{prop8}, \ref{prop9} and \ref{prop10}, we know that is, if $R_0>1$, the epidemic occurs for all traits $x \in \Omega$ no matter how small the initial size of the infecteds is, and its final size distribution is estimated from below by the final size distribution $V$.  On the other hand, the final size goes to zero if the initial size of infecteds goes to zero under the condition $R_0 \le  1$.

\section{The effective reproduction number, control and HIT}

In order to consider the herd immunity threshold and to estimate the effect of preventive control, we introduce here the idea of the effective reproduction number.
Preventive control that works by reducing the availability of susceptibles in the host population is called {\bf S-control}, while the {\bf I-control} means that the preventive control acts primarily on the infectivity of the host individuals \cite{Heesterbeek2006}.

Let $Z:=L^\infty(\Omega) \subset Y$.
For $\phi \in Z_+$, we define the {\bf effective next generation operator at state $\phi \in Z_+$} for S-control by $T_S[\phi]: X_+ \to X_+$:
\begin{equation}
(T_{S}[\phi] \psi)(x):=N(x)e^{-\phi(x)}\int_{\Omega} \Theta(x, \sigma)\psi(\sigma) d\sigma, ~\psi \in X,
\end{equation}
and for I-control by $T_I[\phi]: X_+ \to X_+$:
\begin{equation}
(T_{I}[\phi] \psi)(x):=N(x)\int_{\Omega} \Theta(x, \sigma)e^{-\phi(\sigma)}\psi(\sigma) d\sigma, ~\psi \in X,
\end{equation}
which are both bounded linear operators in $X$. From our assumption \ref{ass2}, both operators are compact and nonsupporting, and $T_S[0]=T_I[0]$ is none other than the next generation operator.

The biological meaning of the {\it state} $\phi$ in the S-control is that $e^{-\phi(x)}$ is the reduction of susceptibility, or the fraction $1-e^{-\phi(x)}$ is immunized by vaccination or by natural infection, so the susceptible density $N(x)$ is replaced by $N(x)e^{-\phi(x)}$ in the next generation operator $T$ to get $T_S$.  On the other hand, in the I-control, $e^{-\phi(x)}$ is the reduction of infectivity, or $1-e^{-\phi(x)}$ is removed from the infectious state by vaccination or by non-pharmaceutical intervention.

The spectral radius of $T_S[\phi]$ ($T_I[\phi]$) is defined as the {\bf effective reproduction number} in the state $\phi$ for S-control (I-control), denoted by $R_{S}[\phi]$ ($R_{I}[\phi]$), that is, 
\begin{equation}
R_S[\phi]=r(T_S[\phi]), \quad R_I[\phi]=r(T_I[\phi]),
\end{equation}
Note that $R_S[0]=R_I[0]$ is none other than the basic reproduction number $R_0$.  

\begin{proposition}\label{prop11}
If $\phi \in Z_+$, it holds that $R_S[\phi]=R_I[\phi]$.
\end{proposition}
\begin{proof}
Let $U[\phi]: X_+ \to X_+$ be a positive operator defined by $(U[\phi]\psi)(x)=e^{-\phi(x)}\psi(x)$.  Then if $\phi \in Z_+$, then $U^{-1}$ is a bounded linear operator, and it follows that $T_S[\phi]=UT_I[\phi]U^{-1}$, that is, $T_S[\phi]$ and $T_I[\phi]$ are similar operators to each other. 
Then $R_S[\phi]=R_I[\phi]$ holds.
\end{proof}

If the vaccination coverage is given by $\epsilon(x)$, then the susceptible population is $(1-\epsilon(x))N(x)$. Let the state $\phi$ be given by  $\phi(x)=-\log (1-\epsilon(x))$, then $\epsilon$ is the vaccination coverage, and the partially immunized host population has the effective reproduction number given by $R_S[\phi]$.

Let $\Psi'[\phi]$ be the Fr\'echet derivative of $\Psi$ at $\phi \in Z_+$.  Then it holds that
\begin{equation}\label{63}
(\Psi'[\phi]y)(x)=\int_\Omega \Theta(x,\sigma)N(\sigma)e^{-\phi(\sigma)}y(\sigma)d\sigma, \quad y \in Z_+.
\end{equation}

\begin{proposition}\label{prop12}
For the effective reproduction number at state $\phi \in Z_+$, it follows that
\begin{equation}
R_S[\phi]=R_I[\phi]=r(\Psi'[\phi]).
\end{equation}
\end{proposition}
\begin{proof}
From proposition \ref{prop11}, we have $R_S[\phi]=R_I[\phi]$. Let $P_\phi$ be a bounded linear operator from $Y$ to $X$ given by $(P_\phi y)(x)=N(x)e^{-\phi(x)}y(x)$.
Then $\Psi'[\phi]$ is a similar operator of $T_S[\phi]$, that is, 
\begin{equation}\label{similar}
P_{\phi}^{-1}T_S[\phi] P_{\phi}=\Psi'[\phi].
\end{equation}
Therefore we have $R_S[\phi]=r(T_S[\phi])=r(\Psi'[\phi])$. 
\end{proof}

If $W(t,x)$ is the cumulative FOI of the renewal equation \eqref{sir*},  then
\begin{equation}
\begin{aligned}
(T_S[W(t,\cdot)] \psi)(x)
&=N(x)\exp(-W(t,x))\int_{\Omega} \Theta(x, \sigma)\psi(\sigma) d\sigma \cr
&=S(t,x)\int_{\Omega} \Theta(x, \sigma)\psi(\sigma) d\sigma.
\end{aligned}
\end{equation}
Then $R_S[W(t,\cdot)]$ gives the effective reproduction number at time $t$ in the natural epidemic process described by \eqref{sir*}.
In particular, we call $R_S[W(\infty,\cdot)]$ the {\bf final reproduction number}, which gives the reproduction number for the remaining susceptible population after the end of the epidemics.

\begin{proposition}\label{prop13}
Let $W(t,x)$ be the cumulative FOI for \eqref{sir*}.  Then $R_S[W(t,\cdot)]$ is monotonically decreasing with respect to $t$ and $R_S[W(\infty,\cdot)]<1$, and there exists a unique $t^*>0$ such that $R_S[W(t^*,\cdot)]=1$ if $R_0>1$.
\end{proposition}
\begin{proof}
Note that $R_S[W(t,\cdot)]$ is monotonically decreasing with respect to $t$, since $W(t,\cdot)$ is monotonically increasing. 
It follows from $1-e^{-x} \ge e^{-x}x$ for $x \ge 0$ that
\begin{equation}\label{66}
\begin{aligned}
W(\infty,x)&=\int_\Omega \Theta(x,\sigma)N(\sigma)(1-e^{-W(\infty,\sigma)})d\sigma \cr
&\ge \int_\Omega \Theta(x,\sigma)N(\sigma)e^{-W(\infty,\sigma)}W(\infty,\sigma)d\sigma\cr
&=(\Psi'[W(\infty,\cdot)] W(\infty,\cdot))(x).
\end{aligned}
\end{equation}
From \eqref{66}, we obtain an inequality in $X$:
\begin{equation}\label{67}
(P_{\phi}W(\infty,\cdot))(x) \ge (T_S[\phi] P_{\phi}W(\infty,\cdot))(x).
\end{equation}
Let $f^*_{\phi} \in X^*_+$ be the adjoint eigenfunctional of $T_S[\phi]$ associated with the eigenvalue $R_S[W(\infty,\cdot)]=r(\Psi'[W(\infty,\cdot)])=r(T_S[\phi])$.  Taking the duality pairing in \eqref{67}, it follows that $\langle f^*_{\phi}, P_{\phi}W(\infty,\cdot)\rangle \ge R_S[W(\infty,\cdot)] \langle f^*_{\phi}, P_{\phi}W(\infty,\cdot)\rangle$, where the equality does not hold because $W(\infty,\cdot) \neq 0$.  Then we know that $R_S[W(\infty,\cdot)] <1$.  If $R_0=R_S[0]>1$, there exists a unique $t^*>0$ such that $R_S[W(t^*,\cdot)]=1$, because $r(\Psi'[W(t,\cdot)])$ is continuous, monotone decreasing with respect to $t$, and it moves from $R_0>1$ to $R_S[W(\infty,\cdot)]<1$ as $t$ goes from 0 to $\infty$.
\end{proof}

Traditionally, the {\bf herd immunity threshold (HIT)} is defined as the proportion of the host population that must be immune to prevent an outbreak. When the HIT is reached by mass vaccination, it can be called the {\bf critical coverage of immunization (CCI)}.

Define a subset $\mathcal W \subset Z_+$ by
\begin{equation}
\mathcal W:=\left\{\phi \in Z_+: R_S[\phi]=R_I[\phi]=1\right\}.
\end{equation}
and let $\epsilon(x)=1-e^{-\phi(x)}$ for $\phi \in \mathcal W$, then $\epsilon$ gives the critical coverage of intervention for S-control, or for I-control.  For the S-control, the HIT for $\phi \in \mathcal W$ is calculated as  
\begin{equation}
{\rm HIT[\phi]}=1-\frac{1}{N}\int_\Omega N(x)e^{-\phi(x)}dx =\int_\Omega \epsilon(x)\omega(x)dx,
\end{equation}
where $\omega(x)=N(x)/\int_{\Omega}N(\sigma)d\sigma$ denotes the trait profile of the host population.

That is, the HIT is not uniquely determined for the heterogeneous population.
As is seen in Proposition \ref{prop13}, for the natural epidemic process described by \eqref{sir*}, if $R_0>1$, there exists a unique time $t^*$ such that $W(t^*,\cdot) \in \mathcal W$.
 Then the HIT is reached at time $t^*$, but the time $t^*$ and the HIT for the natural infection generally depend on the initial data.
For the I-control, for $\phi \in \mathcal W$, $e^{-\phi(x)}$ gives the critical reduction level of infectivity.

\section{Results of the separable mixing assumption}

Here we consider the results of the separable mixing assumption:
\begin{equation}\label{separable}
A(\tau,x,\sigma)=a(x)b(\tau)c(\sigma), \quad H(t,x)=a(x)Q(t), 
\end{equation}
where we assume that $Q$ is a bounded continuous positive function on $\mathbb R_+$ and that $\lim_{t \to \infty}Q(t)=Q(\infty)>0$ exists.
Then we can concretely compute the basic indices defined so far.

\subsection{$R_S[\phi]$ and HIT}

From \eqref{63} and \eqref{separable}, we easily obtain expressions:
\begin{equation}
(\Psi'[\phi]y)(x)=L a(x) \int_\Omega c(\sigma)N(\sigma)e^{-\phi(\sigma)}y(\sigma)d\sigma, ~ y \in Y_+.
\end{equation}
\begin{equation}
\begin{aligned}
R_S[\phi]&=r(\Psi'[\phi])= L\int_\Omega c(\sigma)N(\sigma)e^{-\phi(\sigma)}a(\sigma)d\sigma \cr
&=R_0 \int_{\Omega}\kappa(x)e^{-\phi(x)} dx,
\end{aligned}
\end{equation}
where
\begin{equation}
\kappa(x):=\frac{a(x)c(x)N(x)}{\int_{\Omega}a(\sigma)c(\sigma)N(\sigma)d\sigma},
\end{equation}
Then it follows that
\begin{equation}
\begin{aligned}
\mathcal W&=\left\{\phi \in Z_+:  L\int_\Omega c(\sigma)N(\sigma)e^{-\phi(\sigma)}a(\sigma)d\sigma=1 \right\}.\cr
&=\left\{\phi \in Z_+: \int_{\Omega} \kappa(x)e^{-\phi(x)}dx = \frac{1}{R_0}\right\}.
\end{aligned}
\end{equation}
Note that $\mathcal W \neq \emptyset$, because a constant function 
\begin{equation}
\phi(x)=\log\left(L\int_\Omega c(\sigma)N(\sigma)a(\sigma)d\sigma \right)=\log R_0,
\end{equation}
 belongs to $\mathcal W$.  In this constant case, the HIT (or CCI) is given by
 \begin{equation}
 1-e^{-\phi}=1-\frac{1}{R_0},
 \end{equation}
which traditional result implies that the epidemic will not occur if the fraction greater than $1-1/R_0$ is immunized in the susceptible host population.

\subsubsection*{{\bf Example 2}}
It is easy to see that increasing the immunization fraction in populations with larger $\kappa(x)$ is more effective to lower the effective reproduction number. Let us show an example.
Again let us assume that $\Omega=[0,\infty)$, $a(x)=x$ and $c$ is constant.
Then we have
\begin{equation}
\kappa(x)=\frac{x \omega(x)}{\int_{0}^{\infty}x \omega(x)dx}. 
\end{equation}
  Suppose that
\begin{equation}\label{7.18}
\int_{0}^{\infty}(\epsilon(x)-\langle \epsilon, \omega \rangle)(x-\langle x, \omega \rangle)\omega(x)dx \ge 0,
\end{equation}
where $\langle f, \omega \rangle:=\int_{0}^{\infty}f(x)\omega(x)dx$ denotes the mean of $f$.
Then we obtain 
\begin{equation}\label{7.19}
\langle \epsilon, \kappa \rangle \ge \langle \epsilon, \omega \rangle,
\end{equation}
because if we expand \eqref{7.18}, we get $\langle \epsilon x, \omega \rangle- \langle \epsilon, \omega \rangle \langle x, \omega \rangle \ge 0$. It follows from $\langle \epsilon x, \omega \rangle/\langle x, \omega \rangle=\langle \epsilon, \kappa  \rangle$ that $\langle \epsilon, \kappa \rangle \ge \langle  \epsilon, \omega \rangle$. Therefore, if $\phi^*(x)=-\log (1-\epsilon^*(x)) \in \mathcal W$, we have
\begin{equation}\label{7.20}
\langle \epsilon^*, \kappa \rangle =1-\frac{1}{R_0}\ge \langle \epsilon^*, \omega \rangle.
\end{equation}
The condition \eqref{7.18} and its result \eqref{7.20} imply that if the group with above-average susceptibility is more immunized than the  average immunization rate, the critical immunization fraction of the population $\langle \epsilon^*, \omega \rangle$ is less than $1-1/R_0=\langle \epsilon^*, \kappa \rangle$, which is the HIT for the homogeneous case. It is intuitively obvious that selective immunization of susceptible groups will achieve herd immunity by immunizing a smaller percentage of the population.

\subsection{HIT in the natural epidemics}

In order to calculate the HIT in the natiral epidemic, observe that \eqref{sir*} is written as
\begin{equation}
W(t,x)=a(x)Q(t)+ a(x)\int_\Omega \int_{0}^{t}b(\tau)c(\sigma)N(\sigma)(1-e^{-W(t-\tau,\sigma)})d\tau d\sigma.
\end{equation}
Therefore, there exists a function $w(t)$ such that $W(t,a)=a(x)w(t)$ and
\begin{equation}\label{w}
w(t)=Q(t)+ \int_\Omega \int_{0}^{t}b(\tau)c(\sigma)N(\sigma)(1-e^{-a(\sigma)w(t-\tau)})d\tau d\sigma.
\end{equation}
From Proposition \ref{prop9}, we have
\begin{equation}\label{7.6}
R_S[W(t,\cdot)]= r(\Psi'[a(\cdot)w(t)])=L\int_{\Omega}c(\sigma)a(\sigma)e^{-a(\sigma)w(t)}N(\sigma)d\sigma.
\end{equation}
In this case, $w(t^*)$ satisfies 
\begin{equation}
R_S[W(t^*,\cdot)]=R_S[a(\cdot)w(t^*)]=1,
\end{equation}
 is uniquely determined as the unique positive root $\lambda$ of the characteristic equation
\begin{equation}
L  \int_{\Omega}c(\sigma)a(\sigma)e^{-a(\sigma)\lambda}N(\sigma)d\sigma=1.
 \end{equation}

Since $\lambda$ is independent of the initial data, the HIT in the natural epidemics is calculated as
\begin{equation}
{\rm HIT}[a\lambda]=1-\int_{\Omega}\omega(x)e^{-a(x)\lambda}dx,
\end{equation}
and it is also independent from the initial data, although the time $t^*$ such that $\lambda=w(t^*)$ depends on the initial data $Q$.

\subsection*{{\bf Example 3}}
Suppose that $\Omega=\mathbb R_+$, $a(x)=x$, $c(\sigma)=c$ is constant and the trait profile $\omega$ is given by the gamma function with mean $1$ and  variance $1/p$:
\begin{equation}\label{gamma}
\omega(x)=\frac{p^p}{\Gamma(p)}x^{p-1}e^{-x p}.
\end{equation}
Using the formula \eqref{7.6}, we can compute the effective reproduction number as follows
\begin{equation}\label{Re}
\begin{aligned}
R_S[W(t,\cdot)]&=NLc \int_0^{\infty}\omega(\sigma)\sigma e^{-\sigma w(t)}d\sigma \cr
&=\frac{NLc p^{p+1}}{(p+w(t))^{p+1}}=R_0 s(t)^{1+\frac{1}{p}}, 
\end{aligned}
\end{equation}
where
\[s(t):=\frac{1}{N}\int_{0}^{\infty}S(t,x)dx,\]
denotes the susceptible fraction at time $t$, and it is calculated as 
\begin{equation}
s(t)=\left(\frac{p}{p+w(t)}\right)^p.
\end{equation}
Then the HIT in the natural epidemics is given by 
\begin{equation}
1-s(t^*)=1-R_0^{-\frac{p}{p+1}},
\end{equation}
 which may be far from the homogeneous case of $p=\infty$.  These power-law results have been used in many studies related to COVID-19 \cite{Bootsma2023b, Diekmann2023, Montalban2022, Tkachenko2021a}.

\subsection{Final size}

In the renewal equation \eqref{w}, if we let $t \to \infty$, we obtain an equation for $w(\infty)$:
\begin{equation}\label{winfty}
w(\infty)=Q(\infty)+L\int_{\Omega}c(\sigma)N(\sigma)(1-e^{-a(\sigma)w(\infty)})d\sigma.
\end{equation}
  Then the final susceptible distribution is given by
\begin{equation}
S_\infty(x)=N(x)e^{-a(x)w(\infty)},
\end{equation}
and the attack rate (final size distribution) is calculated as
\begin{equation}
1-\frac{S_\infty(x)}{N(x)}=1-e^{-a(x)w(\infty)}
\end{equation}

If $R_0>1$,  
\begin{equation}
F(x):=L\int_{\Omega}c(\sigma)N(\sigma)(1-e^{-a(\sigma)x})d\sigma.
\end{equation}
has a unique fixed point $v>0$ and $w(\infty) \ge v$.   Then the final size is estimated from below as
\begin{equation}
1-\frac{S_\infty(x)}{N(x)} \ge 1-e^{-a(x)v}.
\end{equation}

\subsection{I-control and epidemic resurgence}

 As was seen above, the final reproduction number $R_e(\infty)$ is less than unity, the epidemic resurgence does not occur after the natural eradication of the disease.  
However, if the infectivity is reduced by non-pharmaceutical intervention (such as universal masking and social distance), after the end of the controlled outbreak, there could remain susceptible population whose effective reproduction number is greater than one.  Then if the intervention is lifted after the end of the first epidemic wave,  the resurgence of the epidemic could occur.  
To formulate the condition for the resurgence, we calculate the critical level of the control for the infectivity.

Again we adopt the separable mixing assumption.
Let $S_\infty[\phi](x)=N(x)e^{-w_\phi(\infty) a(x)}$ be the final susceptible distribution under the I-control with state $\phi$, 
where $w_\phi(\infty)$ is the positive root of the final size equation
\begin{equation}\label{7.27}
w_\phi(\infty)=Q(\infty)+L\int_{\Omega}c(x)e^{-\phi(x)}N(x)(1-e^{-w_\phi(\infty) a(x)})dx.
\end{equation}
Then the root $w_\phi(\infty)$ is non-increasing with respect to $\phi$. 

The effective next generation operator after lifting the prevention policy is given by
\begin{equation}
(T_{w_\phi(\infty) a}\phi)(x)=La(x)N(x)e^{-w_\phi(\infty) a(x)}\langle c, \phi \rangle.
\end{equation}
and the effective reproduction number is
\begin{equation}\label{96}
r(T_{w_\phi(\infty) a})=L\int_{\Omega}c(\sigma)N(\sigma)e^{-w_{\phi}(\infty)a(\sigma)}a(\sigma)d\sigma.
\end{equation}
Then we can define the set of critical I-control by
\begin{equation}
\Sigma:=\left\{\phi \in Z_+: L\int_{\Omega}c(\sigma)N(\sigma)e^{-w_\phi(\infty)a(\sigma)}a(\sigma)d\sigma=1\right\}.
\end{equation}

Let $\lambda^*$ be the unique positive root of the characteristic equation:
\begin{equation}\label{98}
L\int_{\Omega}c(\sigma)N(\sigma)e^{-\lambda^* a(\sigma)}a(\sigma)d\sigma=1,
\end{equation}
whose existence is guaranteed by the assumption $R_0=L\int_{\Omega}c(x)N(x)a(x)dx>1$.  So, we know that
\begin{equation}
\begin{aligned}
\Sigma&=\{\phi \in Z_+: w_\phi(\infty)=\lambda^* \} \cr
&=\left\{\phi: \lambda^*=Q(\infty)+L\int_{\Omega}c(x)e^{-\phi(x)}N(x)(1-e^{-\lambda^* a(x)})dx \right\}.
\end{aligned}
\end{equation}

If $\lambda^*>Q(\infty)$,  $\Sigma$ is not empty.  In fact, if assume that $\phi$ is a constant function $\phi(x)=\phi^*$, we can easily see that
\begin{equation}\label{100}
\phi^*=-\log \left(\frac{\lambda^*-Q(\infty)}{L\int_{\Omega}c(x)N(x)(1-e^{-\lambda^* a(x)})dx}\right) \in \Sigma.
\end{equation}
If the I-control is so strong as $\phi(x)>\phi^*$ for all $x \in \Omega$, the resurgence occurs if the prevention policy is lifted after the previous epidemics have ended.

\subsubsection*{{\bf Example 4}}
In general, it is difficult to calculate the threshold value $\lambda^*$ explicitly.  However, this is not the case if we consider the limiting epidemics described by the homogeneous model with constant parameters.
Assume that the host population is homogeneous, all parameters are constant, so assume that
$a(x)=\beta$, $b(\tau)=e^{-\gamma \tau}$ and $c=1$.
 From \eqref{98}, we have 
 \begin{equation}
 L\int_{\Omega}c(\sigma)N(\sigma)e^{-\lambda^* a(\sigma)}a(\sigma)d\sigma=R_0e^{-\lambda^*\beta}=1,
\end{equation}
hence we have $\lambda^*=(\log R_0)/\beta$.
From \eqref{100} and assuming that $Q(\infty)=0$, it follows that
\begin{equation}\label{101}
\phi^*=-\log \left(\frac{\lambda^*}{\frac{N}{\gamma}\int_{\Omega}\omega(x)(1-e^{-\lambda^* \beta})dx}\right) =-\log\left(\frac{\log R_0}{R_0-1} \right).
\end{equation}
In the homogeneous model with constant parameters, if the transmission control is so strong that $\phi>\phi^*$, and $Q$ is negligibly small (that is, the initial size of the infected is so small),  the supercritical size of the susceptible population remains after the end of the previous epidemic wave, and a resurgence can occur.

\section{Final remarks}\label{sec12}

In the classical theory of the Kermack MacKendrick model, the  basic epidemiological concepts as the basic reproduction number, the effective reproduction number, the herd immunity threshold and the final size can be mathematically formulated by using the renewal integral equation.  The threshold principle tells us that an outbreak occurs in a totally susceptible population if and only if $R_0>1$ and the lower bound of the final size is given as the unique root of the final size equation.   In this paper, we have extended these basic epidemiological concepts and results to allow the theory to account for individual static heterogeneity.  
Our main results provide a mathematical basis for the recent COVID-19-inspired arguments and calculations based on the heterogeneous Kermack-McKendrick model. 

Among the COVID-19-related studies, the most notable ones were that individual variation in susceptibility lowers the herd immunity threshold, and it was shown that the infinite-dimensional compartment model with individual heterogeneity described by the gamma distribution can be transformed into a finite-dimensional ordinary differential equation system with power-law interaction term.

To justify these arguments, we have to deal with unbounded susceptibility and infectivity with non-compact domain.  Our assumption \ref{ass1} is sufficient to prove the well-posedness of the heterogeneous model with unbounded parameters.  On the other hand, our assumption \ref{ass2}  still excludes the unbounded infectivity, so it is a future challenge to construct a theory to deal with more general parameters.  Moreover, it would also be necessary to consider the dynamic heterogeneity to understand the real epidemic waves.

\bmhead{Acknowledgments}
This work was supported by JSPS KAKENHI Grant Number 22K03433 and Japan Agency for Medical Research and Development (JP23fk0108685).

\section*{Declarations}
Not applicable

\bibliography{sn-bibliography}


\end{document}